\documentclass[12pt,letterpaper]{article}
\usepackage{amsmath,amsfonts,amsthm,amssymb,stmaryrd,bm}
\usepackage{graphicx,relsize,epic}

\theoremstyle{plain}

\numberwithin{equation}{section}
\newtheorem{thm}{Theorem}[section]
\newtheorem{lem}[thm]{Lemma}

\newcommand{\positive}{{\mathbb N}}
\newcommand{\complex}{{\mathbb C}}
\newcommand{\real}{{\mathbb R}}
\newcommand{\ascript}{{\mathcal A}}
\newcommand{\cscript}{{\mathcal C}}
\newcommand{\dscript}{{\mathcal D}}
\newcommand{\pscript}{{\mathcal P}}
\newcommand{\qscript}{{\mathcal Q}}
\newcommand{\sscript}{{\mathcal S}}
\newcommand{\rmcyl}{\mathrm{cyl}}
\newcommand{\ahat}{\widehat{a}}
\newcommand{\xhat}{\widehat{x}}
\newcommand{\yhat}{\widehat{y}}
\newcommand{\zhat}{\widehat{z}}
\newcommand{\omegahat}{\widehat{\omega}}
\newcommand{\rhohat}{\widehat{\rho}}
\newcommand{\atilde}{\widetilde{a}}
\newcommand{\abar}{\overline{a}}
\newcommand{\ouparrow}{\mathord{\uparrow}}

\newcommand{\ab}[1]{\left|#1\right|}
\newcommand{\brac}[1]{\left\{#1\right\}}
\newcommand{\paren}[1]{\left(#1\right)}
\newcommand{\sqbrac}[1]{\left[#1\right]}
\newcommand{\elbows}[1]{{\left\langle#1\right\rangle}}
\newcommand{\ket}[1]{{\left|#1\right>}}
\newcommand{\bra}[1]{{\left<#1\right|}}

\begin{document}

\title{LABELED CAUSETS IN\\DISCRETE QUANTUM GRAVITY
}
\author{S. Gudder\\ Department of Mathematics\\
University of Denver\\ Denver, Colorado 80208, U.S.A.\\
sgudder@du.edu
}
\date{}
\maketitle

\begin{abstract}
We point out that labeled causets have a much simpler structure than unlabeled causets. For example, labeled causets can be uniquely specified by a sequence of integers. Moreover, each labeled causet processes a unique predecessor and hence has a unique history. Our main result shows that an arbitrary quantum sequential growth process (QSGP) on the set of labeled causets ``compresses'' in a natural way onto a QSGP on the set of unlabeled causets. The price we have to pay is that this procedure causes an ``explosion'' of values due to multiplicities. We also observe that this procedure is not reversible. This indicates that although many QSGPs on the set of unlabeled causets can be constructed using this method, not all can, so it is not completely general. We close by showing that a natural metric can be defined on labeled and unlabeled causets and on their paths.
\end{abstract}

\section{Introduction}  
In the causal set approach to discrete quantum gravity, a causal set (causet) represents a possible universe at a certain time instant and a possible ``completed'' universe is represented by a path of growing causets \cite{gud11,hen09,rs00,sor03,sur11}. Just as covariance dictates that the laws of physics are independent of the coordinate system employed, in the discrete theory, covariance implies that order isomorphic causets should be identified. That is, a causet should be independent of labeling. This is unfortunate because it is very convenient to work with labeled causets. At a fundamental level, the labeling specifies the ``birth order'' of the vertices of a causet. Although covariance dictates that a causet should be independent of the birth order of its vertices, this order is useful in keeping track of the causets. In particular, the labeled offspring of a labeled causet possess a natural lexicographic total order. This lexicographic order together with its level uniquely specify a labeled causet in terms of a sequence of positive integers.

Although there are many more labeled causets than unlabeled causets, their graph structure is much simpler. This is because the graph of labeled causets forms a tree which implies that each labeled causet has a unique producer (predecessor) and hence possesses a unique history. This unique history structure makes it much simpler to construct a candidate quantum dynamics to form a quantum sequential growth process (QSGP) which is the basis for this approach to discrete quantum gravity.

Even though the graphic structure $(\pscript ,\to )$ of unlabeled causets is quite complicated, the graphic structure $(\pscript ',\to )$ of labeled causets is simple. A QSGP consists of a sequence $\brac{\rho _n}$ of Hilbert space operators, and it is more straightforward to construct these operators for
$\pscript '$ then for $\pscript$. The main point of this paper is that even though the QSGP $\brac{\rho _n}$ for $\pscript '$ can be arbitrary and need not satisfy any covariance conditions, we can ``compress'' $\brac{\rho _n}$ to form a QSGP $\brac{\rhohat _n}$ on $\pscript$ that is automatically covariant. We thus obtain the surprising result that any quantum dynamics on $\pscript '$ compresses in a natural way to a quantum dynamics on
$\pscript$.

As one might suspect, there is a price to be paid for this fortunate circumstance. Since there are usually many ways to label an unlabeled causet, the compression map is many-to-one which results in a multiplicity factor. This factor increases with $n$ and may affect the convergence of quantum measures. In this way, events in $\pscript '$with finite quantum measure may have corresponding events in
$\pscript$ with infinite quantum measure. We also observe that this procedure is not reversible. That is, a QSGP on $\pscript$ may not be ``expandable'' to a QSGP on $\pscript '$. This indicates that although many QSGPs on $\pscript$ can be constructed using this method, not all can, so it is not completely general. In the last section of this article we show that a natural metric can be defined on labeled and unlabeled causets and on their paths.

\section{Quantum Sequential Growth Processes} 
A finite partially ordered set is called a \textit{causet}. In this section we treat only unlabeled causets and two isomorphic causets are considered to be identical. Let $\pscript _n$ be the collection of all causets of cardinality $n$, $n=1,2,\ldots$, and let $\pscript =\cup\pscript _n$ be the collection of all causets. If $a,b$ are elements of a causet $x$, we interpret the order $a<b$ as meaning that $b$ is in the causal future of $a$. If $a<b$ and there is no $c$ with $a<c<b$, then $a$ is a \textit{parent} of $b$ and $b$ is a \textit{child} of $a$. An element $a\in x$ for $x\in\pscript$ is \textit{maximal} if there is no $b\in x$ with $a<b$. If $x\in\pscript _n$, $y\in\pscript _{n+1}$ then $x$ \textit{produces} $y$ if $y$ is obtained from $x$ by adjoining a single maximal element $a$ to $x$. We then write $x\to y$ and
$y=x\uparrow a$. If $x\to y$, we say that $x$ is a \textit{producer} of $y$ and $y$ is an \textit{offspring} of $x$. Of course, $x$ may produce many offspring and a causet may be the offspring of many producers.

The transitive closure of $\to$ makes $\pscript$ into a partially ordered set itself and we call $(\pscript ,\to )$ the \textit{causet growth process} (CGP). A \textit{path} in $\pscript$ is a sequence (string) $\omega _1\omega _2\cdots$, where $\omega _i\in\pscript$ and
$\omega _i\to\omega _{i+1}$, $i=1,2,\ldots\,$. An $n$-\textit{path} in $\pscript$ is a finite string $\omega _1\omega _2\cdots\omega _n$ where, again $\omega _i\in\pscript _i$ and $\omega _i\to\omega _{i+1}$. We denote the set of paths by $\Omega$ and the set of $n$-paths by $\Omega _n$. If
$\omega =\omega _1\omega _2\cdots\omega _n\in\Omega _n$, we define $(\omega\to )\subseteq\Omega _{n+1}$ by
\begin{equation*}
(\omega\to )=\brac{\omega _1\omega _2\cdots\omega _n\omega _{n+1}\colon\omega _n\to\omega _{n+1}}
\end{equation*}
Thus, $(\omega\to )$ is the set of one-step continuations of $\omega$. If $A\subseteq\Omega _n$ we define
$(A\to )\subseteq\Omega _{n+1}$ by
\begin{equation*}
(A\to )=\cup\brac{(\omega\to )\colon\omega\in A}
\end{equation*}

The set of all paths beginning with $\omega\in\Omega _n$ is called an \textit{elementary cylinder set} and is denoted by $\rmcyl (\omega )$. If
$A\subseteq\Omega _n$, then the \textit{cylinder set} $\rmcyl (A)$ is defined by
\begin{equation*}
\rmcyl (A)=\cup\brac{\rmcyl (\omega )\colon\omega\in A}
\end{equation*}
Using the notation
\begin{equation*}
\cscript (\Omega _n)=\brac{\rmcyl (A)\colon A\subseteq\Omega _n}
\end{equation*}
we see that
\begin{equation*}
\cscript (\Omega _1)\subseteq \cscript (\Omega _2)\subseteq\cdots
\end{equation*}
is an increasing sequence of subalgebras of the \textit{cylinder algebra} $\cscript (\Omega )=\cup \cscript (\Omega _n)$. Letting $\ascript$ be the $\sigma$-algebra generated by $\cscript (\Omega )$, we have that $(\Omega ,\ascript )$ is a measurable space. For $A\subseteq\Omega$, we define the sets
$A^n\subseteq\Omega _n$ by
\begin{equation*}
A^n=\brac{\omega _1\omega _2\cdots\omega _n\colon\omega _1\omega _2\cdots\omega _n\omega _{n+1}\cdots\in A}
\end{equation*}
That is, $A^n$ is the set of $n$-paths that can be continued to a path in $A$. We think of $A^n$ as the $n$-step approximation to $A$. We have that
\begin{equation*}
\rmcyl (A^1)\subseteq\rmcyl (A^2)\subseteq\cdots\subseteq A
\end{equation*}
so that $A\subseteq\cap\rmcyl (A^n)$ but $A\ne\cap\rmcyl (A^n)$ in general, even if $A\in\ascript$.

Let $H_n=L_2(\Omega _n)$ be the $n$-\textit{path Hilbert space} $\complex ^{\Omega _n}$ with the usual inner product
\begin{equation*}
\elbows{f,g}=\sum\brac{\overline{f(\omega )}g(\omega )\colon\omega\in\Omega _n}
\end{equation*}
For $A\subseteq\Omega _n$ the characteristic function $\chi _A\in H_n$ with $\|\chi _A\|=\sqrt{\ab{A}\,}$ where $\ab{A}$ denotes the cardinality of
$A$. In particular, $1_n=\chi _{\Omega _n}$ satisfies $\|1_n\|=\sqrt{\ab{\Omega _n}\,}$. A positive operator $\rho$ on $H_n$ that satisfies $\elbows{\rho 1_n,1_n}=1$ is called a \textit{probability operator}. Corresponding to a probability operator $\rho$ we define the
\textit{decoherence functional} $D_\rho\colon 2^{\Omega _n}\times 2^{\Omega _n}\to\complex$ by
\begin{equation*}
D_\rho (A,B)=\elbows{\rho\chi _B,\chi _A}
\end{equation*}
We interpret $D_\rho (A,B)$ as a measure of the interference between the events $A$, $B$ when the system is described by $\rho$. We also define the $q$--\textit{measure} $\mu _\rho\colon 2^{\Omega _n}\to\real ^+$ by $\mu _\rho (A)=D_\rho (A,A)$ and interpret $\mu _\rho (A)$ as the quantum propensity of the event $A\subseteq\Omega _n$. In general, $\mu _\rho$ is not additive on $2^{\Omega _n}$ so $\mu _\rho$ is not a measure. However, $\mu _\rho$ is \textit{grade}-2 additive \cite{djs10,gud11,hen09,sor94} in the sense that if $A,B,C\in 2^{\Omega _n}$ are mutually disjoint, then
\begin{equation*}
\mu _\rho (A\cup B\cup C)=\mu _\rho (A\cup B)+\mu _\rho (A\cup C)+\mu _\rho (B\cup C)-\mu _\rho (A)-\mu _\rho (B)-\mu _\rho (C)
\end{equation*}

Let $\rho _n$ be a probability operator on $H_n$, $n=1,2,\ldots\,$. We say that the sequence $\brac{\rho _n}$ is \textit{consistent} if
\begin{equation*}
D_{\rho _{n+1}}(A\to ,B\to)=D_{\rho _n}(A,B)
\end{equation*}
for every $A,B\subseteq\Omega _n$. We call a consistent sequence $\brac{\rho _n}$ a \textit{quantum sequential growth process} (QSGP). Now let $\brac{\rho _n}$ be a QSGP and denote the corresponding $q$-measures by $\mu _n$. A set $A\in\ascript$ is \textit{suitable} if $\lim\mu _n(A^n)$ exists (and is finite) in which case we define $\mu (A)=\lim\mu _n(A^n)$. We denote the collection of suitable sets by $\sscript (\Omega )$. It follows from consistency that if $A=\rmcyl (B)$ for $B\subseteq\Omega _n$, then
$\lim\mu _m(A^m)=\mu _n(B)$. Hence, $A\in\sscript (\Omega )$ and $\mu (A)=\mu _n(B)$. We conclude that
$\cscript (\Omega )\subseteq\sscript (\Omega )\subseteq\ascript$ and it can be shown that the inclusions are proper, in general. In a certain sense, $\mu$ is a $q$-measure on $\sscript (\Omega )$ that extends the $q$-measures $\mu _n$. There are physically relevant sets in $\ascript$ that are not in $\cscript (\Omega )$. In this case it is important to know whether such a set $A$ is in $\sscript (\Omega )$ and to find $\mu (A)$. For example, if $\omega\in\Omega$, then $\brac{\omega}=\cap\brac{\omega}^n\in\ascript$ but
$\brac{\omega}\notin\cscript (\Omega )$. Also, the complement
$\brac{\omega}'\notin\cscript (\Omega )$.

We now consider a method for constructing a QSGP. A \textit{transition amplitude} is a map $\atilde\colon\pscript\times\pscript\to\complex$ such that
$\atilde (x,y)=0$ if $x\not\to y$ and $\sum _y\atilde (x,y)=1$ for every $x\in\pscript$. This is similar to a Markov chain except $\atilde (x,y)$ may be complex. The \textit{amplitude process} (AP) corresponding to $\atilde$ is given by the maps $a_n\colon\Omega _n\to\complex$ where
\begin{equation*}
a_n(\omega _1\omega _2\cdots\omega _n)
  =\atilde (\omega _1,\omega _2)\atilde (\omega _2,\omega _3)\cdots\atilde (\omega _{n-1},\omega _n)
\end{equation*}
We can consider $a_n$ to be a vector in $H_n$. Notice that
\begin{equation*}
\elbows{1_n,a_n}=\sum _{\omega\in\Omega _n}a_n(\omega )=1
\end{equation*}
Define the rank 1 positive operator $\rho _n=\ket{a_n}\bra{a_n}$ on $H_n$. Since
\begin{equation*}
\elbows{\rho _n1_n,1_n}=\ab{\elbows{1_n,a_n}}^2=1
\end{equation*}
we conclude that $\rho _n$ is a probability operator.

The corresponding decoherence functional becomes
\begin{align*}
D_n(A,B)&=\elbows{\rho _n\chi _B,\chi _A}=\elbows{\chi _B,a_n}\elbows{a_n,\chi _A}\\
&=\sum _{\omega\in A}\overline{a_n(\omega )}\sum _{\omega\in B}a_n(\omega )
\end{align*}
In particular, for $\omega ,\omega '\in\Omega _n$, $D_n(\omega ,\omega ')=\overline{a_n(\omega )}a_n(\omega ')$ are the matrix elements of
$\rho _n$. The $q$-measure $\mu _n\colon 2^{\Omega _n}\to\real ^+$ becomes
\begin{equation*}
\mu _n(A)=D_n(A,A)=\ab{\sum _{\omega\in A}a_n(\omega )}^2
\end{equation*}
It is shown in \cite{gud13b} that the sequence $\brac{\rho _n}$ is consistent and hence forms a QSGP.

\section{Labeled Causets} 
A \textit{labeling} for a causet $x$ is a bijection $\ell\colon x\to\brac{1,2,\ldots ,\ab{x}}$ such that $a,b\in x$ with $a<b$ implies that
$\ell (a)<\ell (b)$. A \textit{labeled causet} is a pair $(x,\ell )$ where $\ell$ is a labeling of $x$. For simplicity, we frequently write
$x=(x,\ell )$ and call $x$ an $\ell$-\textit{causet}. Two $\ell$-causets $x$ and $y$ are \textit{isomorphic} if there exists a bijection $\phi\colon x\to y$ such that $a<b$ if and only if
$\phi (a)<\phi (b)$ and $\ell\sqbrac{\phi (a)}=\ell (a)$ for every $a\in x$. Isomorphic $\ell$-causets are identified. A given unlabeled causet $x$ can always be labeled. Just take a maximal element $a\in x$ and label it $\ell (a)=\ab{x}$. Remove $a$ from $x$ to form $x\smallsetminus\brac{a}$ and label a maximal element $b\in x\smallsetminus\brac{a}$ by $\ell (b)=\ab{x}-1$. Continue this process until there is only one element $c$ left and label it $\ell (c)=1$. To show that $\ell$ is a labeling of $x$, suppose $d,e\in x$ with $d<e$. Now there is a maximal chain $d<d_1<\cdots d_k<e$ in $x$. By the way $\ell$ was contracted, we have
\begin{equation*}
\ell (d)<\ell (d_1)<\cdots\ell (d_k)<\ell (e)
\end{equation*}
so $\ell (d)<\ell (e)$. Whenever, there is a choice of maximal elements we may obtain a new labeling, so there usually are many ways to label a causet. (There are exceptions, like a chain or antichain.)

We denote the set of $\ell$-causets with cardinality $n$ by $\pscript '_n$ and the set of all $\ell$-causets by $\pscript '=\cup\pscript '_n$. The definitions of Section~2 such as producer, offspring, paths, QSGPs and APs are essentially the same for $\ell$-causets as they were for causets. For $x,y\in\pscript '$ if $y=x\uparrow a$, we always label $a$ with the integer $\ab{x}+1=\ab{y}$. We denote the collection of $n$-paths in $\pscript '$ by $\Omega '_n$ and the collection of paths in $\pscript '$ by $\Omega '$. We define the Hilbert spaces $H'_n=L_2(\Omega '_n)$ as before. If $\omega =\omega _1\omega _2\cdots\omega _n\in\Omega '_n$ we say that $\omega _j$ is \textit{contained} in $\omega$, $j=1,2,\ldots ,n$.

\begin{lem}       
\label{lem31}
An $\ell$-causet $y$ cannot be the offspring of two distinct $\ell$-causet producers.
\end{lem}
\begin{proof}
If we delete the element of $y$ labeled $\ab{y}$ we obtain a producer of $y$. But any producer of $y$ is obtained in this way so there is only one
$\ell$-causet that produces $y$.
\end{proof}

The next lemma shows that unlike in $\pscript$, paths in $\pscript '$ never cross (except at $\omega _1$).

\begin{lem}       
\label{lem32}
If $x\in\pscript '_n$, then $x$ is contained in a unique $n$-path.
\end{lem}
\begin{proof}
Let $\omega _n=x$. If we delete the element of $\omega _n$ labeled $n$, then the resulting set $\omega _{n-1}$ is an $\ell$-causet with $\omega _{n-1}\to\omega _n$ If we next delete the element of $\omega _{n-1}$ labeled $n-1$, then the resulting set
$\omega _{n-1}$ is an $\ell$-causet with $\omega _{n-2}\to\omega _{n-1}$. Continue this process until we obtain the one element
$\ell$-causet $\omega _1$. Then $\omega =\omega _1\omega _2\cdots\omega _n$ is an $n$-path containing $x$. If there were another $n$-path $\omega '=\omega '_1\omega '_2\cdots\omega '_n$ with $\omega '_n=x$, then $\omega '_{n-1}=\omega _{n-1}$ because of Lemma~\ref{lem31}. But by Lemma~\ref{lem31} again, $\omega '_{n-2}=\omega _{n-2}$ and continuing we obtain $\omega '=\omega$. Hence, $\omega$ is unique.
\end{proof}

Let $x=\brac{a_1,a_2,\ldots ,a_n}$ be an $\ell$-causet where we can assume without loss of generality that the label on $a_j$ is $j$, $j=1,\ldots ,n$. Define
\begin{equation*}
j_x\ouparrow =\brac{i\in\positive\colon a_j\le a_i}
\end{equation*}

\begin{lem}       
\label{lem33}
If $x,y,z\in\pscript '$ and $x\to y,z$, then $j_y\ouparrow\subseteq j_z\ouparrow$ or $j_x\ouparrow\subseteq j_y\ouparrow$ for all
$j=1,2,\ldots ,\ab{y}$.
\end{lem}
\begin{proof}
Let $y=\brac{a_1,\ldots ,a_n}$ and $z=\brac{b_1,\ldots ,b_n}$. We can assume that $a_i=b_i$, $i=1,\ldots ,n-1$. If $a_j\not\le a_n$ and
$b_j\not\le b_n$, then $j_y\ouparrow= j_z\ouparrow$. If $a_j\le a_n$ and $b_j\le b_n$, then again $j_y\ouparrow= j_z\ouparrow$. If
$a_j\le a_n$ and $b_j\not\le b_n$ then $j_z\ouparrow\subseteq j_y\ouparrow$ and if $a_j\not\le a_n$ and $b_j\le b_n$ then
$j_y\ouparrow\subseteq j_z\ouparrow$.
\end{proof}

Order the offspring of $x\in\pscript '$ lexicographically as follows. If $x\to y,z$ then $y<z$ if $1_y\ouparrow =1_z\ouparrow ,\ldots $
$j_y\ouparrow = j_z\ouparrow$, $(j+1)_y\ouparrow\subsetneq (j+1)_z$.

\begin{thm}       
\label{thm34}
The relation $<$ is a total order.
\end{thm}
\begin{proof}
Clearly $y\not< y$. If $y<z$, the $z\not < y$. Suppose that $x\to y,z,u$ and $y<z$, $z<u$. Then
$1_y\ouparrow =1_z\ouparrow ,\cdots $, $j_y\ouparrow =j_z\ouparrow$, $(j+1)_y\ouparrow\subsetneq (j+1)z$ and 
$1_z\ouparrow =1_u\ouparrow ,\ldots ,k_z\ouparrow=k_u\ouparrow$, $(k+1)_z\subsetneq (k+1)_u\ouparrow$. We then have
\begin{equation*}
1_y\ouparrow =1_u\ouparrow ,\ldots ,\min (j,k)_y\ouparrow =\min (j,k)_u\ouparrow
\end{equation*}
If $\min (j,k)=j$, then
\begin{equation*}
1_y\ouparrow =1_u\ouparrow ,\ldots ,j_y\ouparrow =j_z\ouparrow =j_u\ouparrow, 
  (j+1)_y\ouparrow\subsetneq (j+1)_z\ouparrow =(j+1)_u\ouparrow
\end{equation*}
If $\min (j,k)=k$, then
\begin{equation*}
1_y\ouparrow =1_u\ouparrow ,\ldots ,k_y\ouparrow =k_z\ouparrow =k_u\ouparrow ,
  (k+1)_y\ouparrow =(k+1)_z\ouparrow\subsetneq (k+1)_u\ouparrow
\end{equation*}
In either case, $y<u$ so $<$ is a partial order relation. To show that $<$ is a total order relation, suppose that $x\to y,z$. If
$j_y\ouparrow =j_z\ouparrow$ for $j=1,2,\ldots ,\ab{y}$, then the adjoined maximal element of $y$ has the same parents as the adjoined maximal element of $z$. Hence, $y$ and $z$ are isomorphic $\ell$-causets so $y=z$. Otherwise, by Lemma~\ref{lem33},
$1_y\ouparrow =1_z\ouparrow ,\ldots ,j_y\ouparrow =j_z\ouparrow$ and $(j+)_y\ouparrow\subsetneq (j+1)_z\ouparrow$ or
$(j+1)_z\ouparrow\subsetneq (j+1)_y\ouparrow$ so $y<z$ or $z<y$.
\end{proof}

Two elements $a,b$ of an $\ell$-causet are \textit{comparable} if $a\le b$ or $b\le a$. Otherwise, $a$ and $b$ are \textit{incomparable}. An \textit{antichain} in $x\in\pscript '$ is a set of mutually incomparable elements of $x$. It is shown in \cite{gud13a} that the number of offspring $o(x)$ of $x\in\pscript '$ is the number of distinct antichains in $x$. Let $y_1<y_2<\cdots <y_{o(x)}$ be the offspring of
$x\in\pscript '$ ordered lexicographically. We call $j$ the \textit{succession} of $y_j$ and $\ab{y_i}$ the \textit{generation} of $y_j$. If
$x\in\pscript '$ with $x\ne\emptyset$, then $x$ has a unique producer so its succession $s(x)$ is well defined. By convention we define $s(\emptyset )=0$. We can uniquely specify each $x\in\pscript '$ by listing its \textit{succession sequence}
\begin{equation*}
\paren{s(x_0),s(x_1),\ldots ,s(x_n)}
\end{equation*}
where $x_0\to x_1\to\cdots\to x_n=x$, $n=\ab{x}$. Of course $x_0=\emptyset$ and $s(x_0)=0$, $s(x_1)=1$ for all $x\in\pscript '$.
\bigskip

\hskip -3pc
\includegraphics[scale=.65]{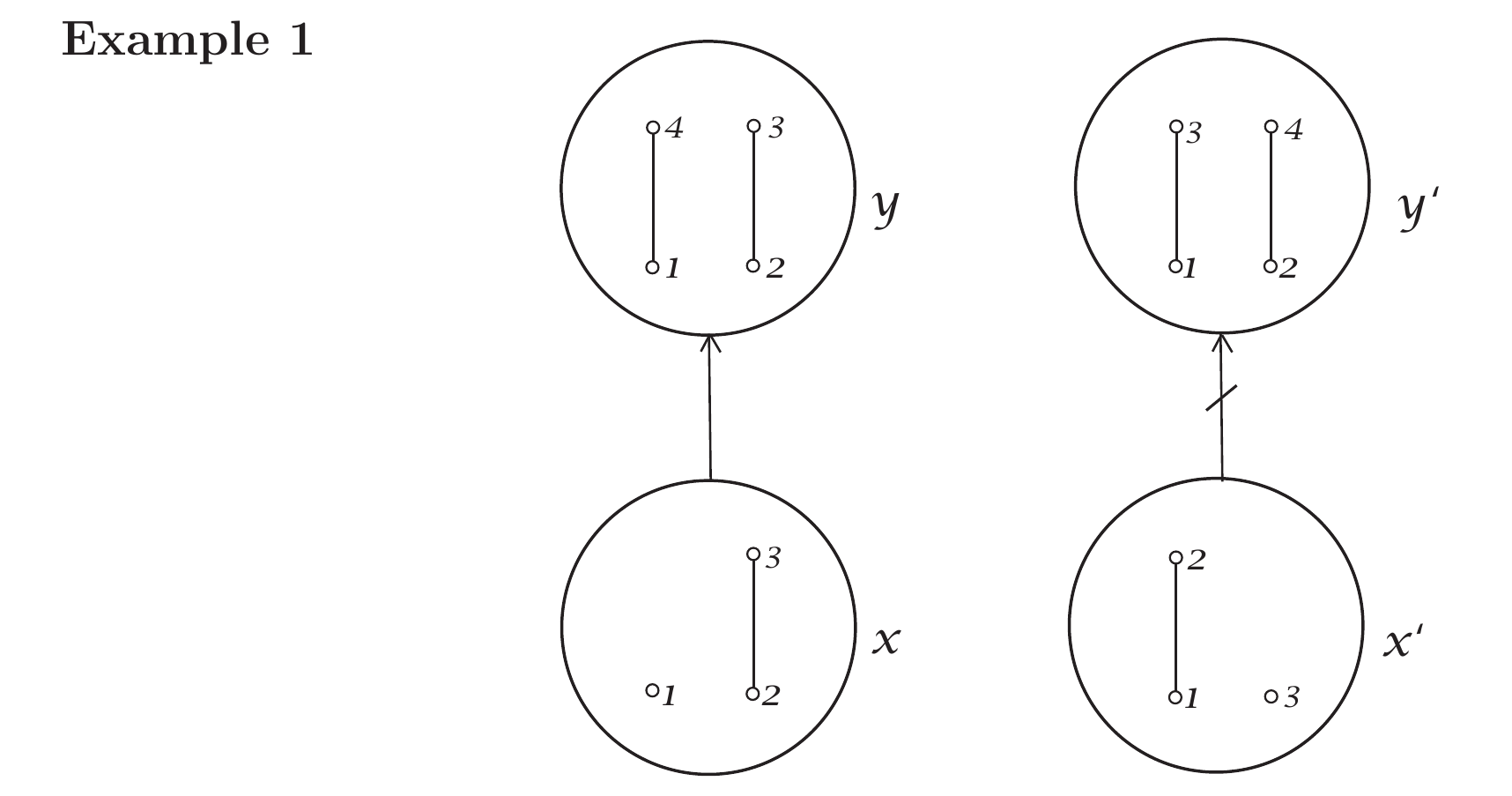}
\medskip

If $x,y\in\pscript '$ we write $x\sim y$ if $x$ and $y$ are order isomorphic. Then $\sim$ is an equivalence relation and we denote the equivalence class containing $x$ by $[x]$. Let
\begin{equation*}
\qscript =\pscript '/\sim =\brac{[x]\colon x\in\pscript '}
\end{equation*}
If $x\in\pscript '$, let $\xhat\in\pscript$ be $x$ without the labels. Since we identify isomorphic causets we have $x\sim y$ if and only if
$\xhat =\yhat$. Letting $\phi\colon\qscript\to\pscript$ be $\phi\paren{[x]}=\xhat$ we see that $\phi$ is well-defined and it is easy to check that $\phi$ is a bijection. If $x,y\in\pscript '$ and $x\to y$, then clearly $\xhat\to\yhat$ in $\pscript$. Unfortunately, if $x\to y$ in $\pscript '$ and $x'\sim x$, $y'\sim y$, then we can have $x'\not\to y'$ as Example~1 shows. 

For $x,y\in\pscript '$ we write $[x]\to [y]$ if there exist $x',y'\in\pscript '$ with $x'\sim x$, $y'\sim y$ and $x'\to y'$. However, the identification $\phi\paren{[x]}=\xhat$ is not very useful because $[x]\to [y]$ implies $\phi\paren{[x]}\to\phi\paren{[y]}$ but
$\phi\paren{[x]}\to\phi\paren{[y]}$ need not imply $x\to y$ as the previous example shows. It is more useful to define a map
$\psi\colon\Omega '_n\to\Omega _n$ as follows. If $\omega =\omega _1\omega _3\cdots\omega _n\in\Omega '_n$ we define
\begin{equation*}
\psi (\omega )=\omegahat =\omegahat _1\omegahat _2\cdots\omegahat _n
\end{equation*}
Since any $x\in\pscript$ can be labeled, we have that $\psi\colon\Omega '_n\to\Omega _n$ is surjective.

\begin{lem}       
\label{lem35}
Let $\omega =\omega _1\omega _2\cdots\omega _n$, $\omega '=\omega '_1\omega '_2\cdots\omega '_n\in\Omega '_n$. Then
$\psi (\omega )=\psi (\omega ')$ if and only if $\omega _j\sim\omega '_j$, $j=1,2,\ldots ,n$.
\end{lem}
\begin{proof}
If $\psi (\omega )=\psi (\omega ')$ then $\omegahat =\omegahat '$ so
\begin{equation*}
\omegahat _1\omegahat _2\cdots\omegahat _n=\omegahat '_1\omegahat '_2\cdots\omegahat '_n
\end{equation*}
which implies that $\omegahat _j=\omegahat '_j$, $j=1,\ldots ,n$. We conclude that $\omega _j\sim\omega '_j$, $j=1,\ldots ,n$. Conversely, if $\omega _j\sim\omega '_j$, $j=1,\ldots ,n$, then $\omegahat _j=\omegahat '_j$. Hence
\begin{equation*}
\psi (\omega )=\omegahat _1\omegahat _2\cdots\omegahat _n=\omegahat '_1\omegahat '_2\cdots\omegahat '_n
  =\psi (\omega ')\qedhere
\end{equation*}
\end{proof}

 Let $x,y\in\pscript '$ with $x\to y'$ and $y'\sim y$ we write $y\sim _xy'$. Then $\sim _x$ is an equivalence relation and we denote the equivalence classes by $[y]_x$. If $x\to y$, the \textit{multiplicity} of $x\to y$ is denoted by $m(x\to y)$ and is defined by
 $m(x\to y)=\ab{[y]_x}$. The \textit{multiplicity} $m(\omega )$ of $\omega\in\Omega '_n$ is defined by
 $m(\omega )=\ab{\psi ^{-1}(\omegahat )}$.

\begin{lem}       
\label{lem36}
If $\omega =\omega _1\omega _2\cdots\omega _n\in\Omega '_n$, then
\begin{equation*}
m(\omega )=m (\omega _1\to\omega _2)m(\omega _2\to\omega _3)\cdots m(\omega _{n-1}\to\omega _n)
\end{equation*}
\end{lem}
\begin{proof}
If $\omega '=\omega '_1\omega '_2\cdots\omega '_n\in\Omega '_n$ where $\omega '_i\sim\omega _i$, $i=1,\ldots ,n$, then
$\psi (\omega ')=\psi (\omega )$. The number of $n$-paths of the form $\omega '$ is
\begin{equation*}
m(\omega _1\to\omega _2)m(\omega _2\to\omega _3)\cdots m(\omega _{n-1}\to\omega _n)
\end{equation*}
The result follows.
\end{proof}

\section{Amplitude Processes on $\pscript '$} 
In this section we present examples of various APs on $\pscript '$. As shown in Section~2, these APs can be used to construct QSGPs on $\pscript '$. Moreover, we shall show in Section~5 that any QSGP on $\pscript '$ can be ``compressed'' to a QSGP on $\pscript$. Recall that a transition amplitude on $\pscript '$ is a map $\atilde\colon\pscript '\times\pscript '\to\complex$ such that $\atilde (x,y)=0$ if $x\not\to y$ and $\sum _y\atilde (x,y)=1$. We say that $\atilde$ is \textit{covariant} if $\atilde (x,y)$ is independent of the labeling of $x$ and $y$; that is, if $x\sim x'$, $y\sim y'$ then $\atilde (x,y)=\atilde (x',y')$

If $y=x\uparrow a$, then $y$ is a \textit{leaf offspring} of $x$ if $a$ has no more than one parent. It is clear that for $x\in\pscript '$, $x$ has$\ab{x}+1$ leaf offspring and that there are $n!$ leaf offspring in $\pscript '_n$. For each $x\in\pscript '$, label the leaf offspring of $x$ lexicographically, $y_1,y_2,\ldots ,y_{\ab{x}+1}$. Define $\atilde\colon\pscript '\times\pscript '\to\complex$ by $\atilde (x,y)=0$ unless $y$ is a leaf offspring of $x$ and then
\begin{equation*}
\atilde (x,x_j)=-e^{2\pi ij(\ab{x}+2)},\quad j=1,\ldots ,\ab{x}+1
\end{equation*}
where of course $i=\sqrt{-1\,}$. It is easy to check that $\atilde$ is a transition amplitude on $\pscript '$ but $\atilde$ is not covariant.

If $\omega =\omega _1\omega _2\cdots\omega _n\in\Omega '_n$, then 
\begin{equation*}
a_n(\omega )=\atilde (\omega _1,\omega _2)\cdots\atilde (\omega _{n-1},\omega _n)
\end{equation*}
Hence, if $\omega _1\omega _2\cdots\in\Omega '$ and at least one $\omega _j$ is not a leaf offspring, then
\begin{equation*}
\mu _n\paren{\brac{\omega}^n}=\ab{a_n(\omega _1\cdots\omega _n)}^2=0
\end{equation*}
for $n$ sufficiently large. Hence, $\brac{\omega}\in\sscript (\Omega ')$ and $\mu\paren{\brac{\omega}}=0$. If every $\omega _n$ is a leaf offspring, then $\mu _n\paren{\brac{\omega}^n}=1$ for all $n$ so $\brac{\omega}\in\sscript (\Omega ')$ with $\mu\paren{\brac{\omega}}=1$. As another example, let $A\subseteq\Omega '$ be the set of paths for which no first succession leaf offspring except $\omega _1$ appear. Letting $F$ be the set of all first succession leaf offspring (except $\omega _1$), we have
\begin{align*}
\mu _n(A^n)&=\ab{\sum\brac{\atilde (\omega _1,\omega _2)\cdots\atilde (\omega _{n-1},\omega _n)\colon\omega _j\notin F}}^2\\
&=\ab{\sum\brac{\atilde (\omega _1,\omega _2)
  \cdots\atilde (\omega _{n-2},\omega _{n-1})(1+e^{2\pi i/(n+1)})\colon\omega _j\notin F}}^2\\
&=\ab{\sum\!\brac{\atilde (\omega _1,\omega _2)
  \cdots\atilde (\omega _{n-3},\omega _{n-2})(1\!+\!e^{2\pi i/n})(1\!+\!e^{2\pi i/(n+1)})\colon\!\omega _j\!\notin\!F}}^2\\
&\vdots\\
&=\ab{(1+e^{2\pi i/3})}^2\ab{(1+e^{2\pi i/4})}^2\cdots\ab{(1+e^{2\pi i/(n+1)}}^2\\
&=(2+2\cos 2\pi /3)^2(2+2\cos\pi /2)^2\cdots (2+2\cos 2\pi /(n+1))^2
\end{align*}
We conclude that $\lim\mu _n(A^n)=\infty$ so $A\not\in\sscript (\Omega )$.

As another example, let $\atilde (x,y)=0$ unless $y$ is a leaf offspring of $x$ and then $\atilde (x,y)=(\ab{x}+1)^{-1}$. Then $\atilde$ is a covariant transition amplitude. Let $L_n\subseteq\Omega '_n$ be the set of $n$-paths $\omega =\omega _1\omega _2\cdots\omega _n$ such that $\omega _j$ are leaf offspring, $j=1,2,\ldots ,n$. Then $a_n(\omega )=0$ for every $\omega\in\Omega '_n\smallsetminus L_n$ and $a_n(\omega )=1/n!$ for every $\omega\in L_n$. We conclude that $D_n(\omega ,\omega ')=0$ unless $\omega ,\omega '\in L_n$ in which case $D_n(\omega ,\omega ')=(n!)^{-2}$. Moreover, for $A,B\subseteq\Omega '_n$ we have
\begin{equation*}
D_n(A,B)=\sum\brac{D_n(\omega ,\omega ')\colon\omega\in A,\omega '\in B}=\frac{\ab{A\cap L_n}\ab{B\cap L_n}}{(n!)^2}
\end{equation*}
and $\mu _n(A)=(n!)^{-2}\ab{A\cap L_n}^2$. Let $L\subseteq\Omega '$ be the set of paths $\omega =\omega _1\omega _2\cdots$ where $\omega _j$ are leaf offspring, $j=1,2,\ldots\,$. If $\omega\in\Omega '\smallsetminus L$, then clearly $\mu _n\paren{\ab{\omega}^n}=0$ so $\brac{\omega}\in\sscript (\Omega ')$ with $\mu\paren{\brac{\omega}}=0$. If $\omega\in L$ then
$\mu _n\paren{\brac{\omega}^n}=(n!)^{-2}\to 0$ so again $\brac{\omega}\in\sscript (\Omega ')$ with $\mu\paren{\brac{\omega}}=0$. Also, if $\omega\in L$, then
\begin{equation*}
\mu _n\paren{\brac{\omega}'^n}=\frac{(n!-1)^2}{(n!)^2}=\paren{1-\tfrac{1}{n!}}^2\to 1
\end{equation*}
so $\brac{\omega}'\in\sscript (\Omega )$ with $\mu\paren{\brac{\omega}'}=1$. In a similar way, if $\omega\in\Omega '\smallsetminus L$ then $\brac{\omega}'\in\sscript (\Omega )$ with $\mu\paren{\brac{\omega}'}=1$. Let $A\subseteq L$ be the subset of $L$ consisting of paths $\omega =\omega _1\omega _2\cdots$ where each $\omega _i$ is a connected graph. It is easy to check that $\ab{A^n}=(n-1)!$. Hence,
\begin{equation*}
\mu _n(A^n)=\frac{\sqbrac{(n-1)!}^2}{(n!)^2}=\frac{1}{n^2}\to 0
\end{equation*}
Thus, $A\in\sscript (\Omega ')$ and $\mu (A)=0$. Since $\mu _n$ is the square of a measure on $\cscript (\Omega '_n)$, $\mu$ has a unique extension $\nu$ from $\cscript (\Omega ')$ to $\ascript '$ as a square of a measure. It follows that $\sscript (\Omega ')=\ascript '$. However, $\mu (A)\ne\nu (A)$ for all $A\in\ascript '$, in general, because we need not have $A=\cap\rmcyl (A^n)$.

We now briefly mention two covariant transition amplitudes that may have physical relevance. The first is complex percolation
$\atilde\colon\pscript '\times\pscript '\to\complex$ \cite{djs10}. As usual $\atilde (x,y)=0$ if $x\not\to y$. Let $p\in\complex$ be arbitrary. If $y=x\uparrow a$ we define $\atilde (x,y)=p ^\pi(1-p)^u$ where $\pi$ is the number of parents of $a$ and $u$ is the number of elements of $x$ that are incomparable to $a$. It is clear that $\atilde$ is covariant. To show that the Markov condition $\sum _y\atilde (x,y)=1$ holds, it is well-known form the classical theory that this condition holds if $0\le p\le 1$. By analytic continuation, the condition still holds for complex $p$.

Our last example is a quantum action dynamics presented in \cite{gud13a}. This dynamics has the form of a discrete Feynman integral. Since this formalism was treated in detail in \cite{gud13a}, we refer the reader to that reference for further consideration.

\section{Compressing a QSGP from $\pscript '$ to $\pscript$} 
This section shows that an arbitrary QSGP on $\pscript '$ can be compressed in a natural way to a QSGP on $\pscript$. The
\textit{compression operator} $S_n\colon H'_n\to H_n$ is the linear operator defined by
\begin{equation*}
S_n\chi _\omega =\chi _{\psi (\omega )}=\chi _{\omegahat}
\end{equation*}
We also define the \textit{covariance operator} $T_n\colon H_n\to H'_n$ as the linear operator given by
\begin{equation*}
T_n\chi _\gamma =\sum\brac{\chi _\omega\in H'_n\colon\omegahat =\gamma}
\end{equation*}

\begin{lem}       
\label{thm51}
The operators $S_n$ and $T_n$ satisfy $T_n=S_n^*$.
\end{lem}
\begin{proof}
For $\omega\in\Omega '_n$ and $\gamma\in\Omega _n$ we have that
\begin{equation*}
\elbows{S_n\chi _\omega ,\chi _\gamma}=\elbows{\chi _{\omegahat},\chi _\gamma}=\delta _{\omegahat ,\gamma}
\end{equation*}
Moreover,
\begin{equation*}
\elbows{\chi _\omega ,T_n\chi _\gamma}=\elbows{\chi _\omega ,\sum\brac{\chi _\omega\in H'_n\colon\omegahat =\gamma}}
  =\delta _{\omegahat ,\gamma}
\end{equation*}
The result now follows.
\end{proof}

The next theorem is the main result of this section.

\begin{thm}       
\label{thm52}
If $\brac{\rho _n}$ is a QSGP for $\pscript '$, then $\rhohat _n=S_n\rho _nS_n^*$ is a QSGP for $\pscript$.
\end{thm}
\begin{proof}
We have that $\rhohat _n$ is positive because
\begin{equation*}
\elbows{\rhohat _nf,f}=\elbows{S_n\rho _nS_n^*f,f}=\elbows{\rho _nS_n^*f,S_n^*f}\ge 0
\end{equation*}
for every $f\in H_n$. The normalization condition follows from
\begin{equation*}
\elbows{\rhohat _n1_n,1_n}=\elbows{\rho _nS_n^*1_n,S_n^*1_n}=\elbows{\rho _nT_n1_n,T_n1_n}=\elbows{\rho _n1_n,1_n}=1
\end{equation*}
Consistency follows from
\begin{align*}
\elbows{\rhohat _{n+1}\chi _{(\omega\to)},\chi _{(\omega '\to)}}
  &=\elbows{\rho _{n+1}T_{n+1}\chi _{(\omega\to )},T_{n+1}\chi _{(\omega '\to )}}\\
  &==\elbows{\rho _nT_n\chi _\omega ,T_n\chi _{\omega '}}=\elbows{\rhohat _n\chi _\omega ,\chi _{\omega '}}
\end{align*}
because
\begin{equation*}
\elbows{\rhohat _{n+1}\chi _{A\to},\chi _{B\to}}=\elbows{\rhohat _n\chi _A,\chi _B}
\end{equation*}
results from bilinearity.
\end{proof}

The compression procedure is not reversible in the sense that if $\rho _n$ is a QSGP for $\pscript$, then $\rho '_n=S_n^*\rho _nS_n$ need not be a QSGP for $\pscript '$. For example, in general
\begin{equation*}
\elbows{\rho '_n1_n,1_n}=\elbows{\rho _nS_n1_n,S_n1_n}\ne 1
\end{equation*}
We now show that the compression operators grow according to multiplicity. First $S_n^*S_n\colon H'_n\to H'_n$ satisfies
\begin{equation*}
S_n^*S_n\chi _\omega =T_n\chi _{\omegahat}=\sum\brac{\chi _{\omega '}\in H'_n\colon\omegahat '=\omegahat}
\end{equation*}
and $S_nS_n^*\colon H_n\to H_n$ satisfies
\begin{equation*}
S_nS_n^*\chi _\gamma =S_nT_n\chi _\gamma
  =S_n\sum\brac{\chi _\omega\in H'_n\colon\omegahat =\gamma}=m(\gamma )\chi _\gamma
\end{equation*}
We thus see that $S_nS_n^*$ is diagonal with eigenvalues $\brac{m(\gamma )\colon\gamma\in\Omega _n}$. Hence,
\begin{equation*}
\|S_nS_n^*\|=\max\brac{m(\gamma )\colon\gamma\in\Omega _n}
\end{equation*}
By the $C^*$-identity, we have
\begin{equation*}
\|T_n\|=\|S_n\|=\sqbrac{\|S_nS_n^*\|}^{1/2}=\sqbrac{\max\brac{m(\gamma )\colon\gamma\in\Omega _n}}^{1/2}
\end{equation*}

Suppose we have an AP $\brac{a_n}$ on $\pscript '$. We have seen that $\rho _n=\ket{a_n}\bra{a_n}$ is a QSGP for $\pscript '$. The corresponding QSGP $\rhohat _n$ for $\pscript$ becomes
\begin{equation*}
\rhohat _n=S_n\ket{a_n}\bra{a_n}S_n^*
\end{equation*}
Hence,
\begin{align*}
\rhohat _n\chi _\gamma&=S_n\ket{a_n}\bra{a_n}T_n\chi _\gamma =S_n\ket{a_n}\elbows{a_n,T_n\chi _\gamma}\\
  &=S_n\ket{a_n}\elbows{a_n,\sum\brac{\chi _\omega\in H'_n\colon\omegahat =\gamma}}\\
  &=S_n\ket{a_n}\sum\brac{\elbows{a_n,\chi _\omega}\colon\omegahat =\gamma}\\
  &=S_n\ket{a_n}\sum\brac{\abar _n(\omega )\colon\omegahat =\gamma}
\end{align*}
We conclude that
\begin{align*}
D_n(\gamma ,\gamma ')&=\elbows{\rhohat _n\chi _{\gamma '},\chi _\gamma}
  =\elbows{S_n\mid a_n}\elbows{a_n\mid T_n\chi _{\gamma '},\chi _\gamma}\\
  &=\overline{\elbows{a_n,T_n\chi _{\gamma '}}}\elbows{a_n,T_n\chi _\gamma}\\
  &=\sum\brac{\overline{a_n(\omega )}\colon\omegahat =\gamma}\sum\brac{a_n(\omega ')\colon\omegahat '=\gamma '}
\end{align*}
It follows that
\begin{equation*}
\mu _n(A)=\ab{\sum _{\gamma\in A}\brac{a_n(\omega )\colon\omegahat =\gamma}}^2
\end{equation*}
and in particular,
\begin{equation*}
\mu _n(\gamma )=\ab{\sum\brac{a_n(\omega )\colon\omegahat =\gamma}}^2
\end{equation*}

We say that an AP $\brac{a_n}$ on $\pscript '$ is \textit{covariant} if $a_n(\omega )$ is independent of the labeling of $\omega$; that is, $a_n(\omega )=a_n(\omega ')$ whenever $\omega\sim\omega '(\omegahat =\omegahat ')$. Of course, for a covariant transition amplitude $\atilde$, the corresponding AP $\brac{a_n}$ is covariant. If the AP $\brac{a_n}$ on $\pscript '$ is covariant, then
\begin{equation*}
D_n(\gamma ,\gamma')=m(\omega )m(\omega ')\abar _n(\omega )a_n(\omega ')
\end{equation*}
where $\omegahat =\gamma$, $\omegahat '=\gamma '$. We also have
\begin{equation*}
\mu _n(A)=\ab{\sum _{\gamma\in A}\brac{m(\omega )a_n(\omega )\colon\omegahat =\gamma}}^2
\end{equation*}
and in particular,
\begin{equation*}
\mu _n(\gamma )=\ab{m(\omega )}^2\ab{a_n(\omega )}^2
\end{equation*}
where $\omegahat =\gamma$. These last three equations exhibit the ``explosion'' in values resulting from multiplicity most clearly.

We can illustrate the situation directly for covariant APs as follows. Let $\atilde\colon\pscript '\times\pscript '\to\complex$ be a covariant transitional amplitude. For $\gamma =\gamma _1\gamma _2\cdots\gamma _n\in\Omega _n$, define
\begin{equation*}
\ahat _n(\gamma )=m(\omega _1\to\omega _2)\atilde (\omega _1,\omega _2)
   \cdots m(\omega _{n-1}\to\omega _n)\atilde (\omega _{n-1},\omega _n)
\end{equation*}
where $\omegahat _i=\gamma _i$, $i=1,\ldots ,n$. Since $\atilde$ is covariant, $\ahat _n(\gamma )$ is well-defined and letting
$\omega =\omega _1\omega _2\cdots\omega _n\in\Omega '_n$ where $\omegahat _i=\gamma _i$, $i=1,\ldots ,n$ we have
$\omegahat =\gamma$ and
\begin{equation*}
\ahat _n(\gamma )=m(\omega )a_n(\omega )
\end{equation*}
Our previous results now follow.

\section{Metrics} 
We have seen in Section~3 that any $x\in\pscript '_n$ is uniquely determined by its succession sequence
\begin{equation*}
s(x)=(s_0,s_1,\ldots ,s_n)
\end{equation*}
If $y\in\pscript '_n$ has succession sequence
\begin{equation*}
s(y)=(t_0,t_1,\ldots ,t_n)
\end{equation*}
we write $x\prec y$ is $s(x)$ precedes  $s(y)$ lexicographically:
\begin{equation*}
s_0=t_0, s_1=t_1,\ldots ,s_j=t_j,s_{j+1}<t_{j+1}
\end{equation*}
for some $j$ with $j\le n-1$. Then $\prec$ is a total order on $\pscript '_n$. We can then well-order the elements of $\pscript '_n$ as
\begin{equation*}
x_1\prec x_2\prec\cdots\prec x_m
\end{equation*}
where $m=\ab{\pscript '_n}$. We now define $\rho (x_i,x_j)=\ab{i-j}$ for all $x_i,x_j\in\pscript '_n$. It is clear that $\rho (x_i,x_j)=0$ if and only if $x_i=x_j$ and that $\rho (x_i,x_j)=\rho (x_j,x_i)$ for all $x_i,x_j\in\pscript '_n$. Moreover, by the triangle inequality for real numbers we have
\begin{equation*}
\rho(x_i,x_j)=\ab{i-j}\le\ab{i-k}+\ab{k-j}=\rho (x_i,x_k)+\rho (x_k,x_j)
\end{equation*}
for all $x_k\in\pscript '_n$. Thus, $\rho$ is a metric on $\pscript '_n$. For $x,y\in\pscript _n$ define
\begin{equation*}
\rho (x,y)=\max\brac{\rho (x',y')\colon\xhat '=x,\yhat '=y}
\end{equation*}
if $x\ne y$ and define $\rho (x,y)=0$, otherwise.

\begin{thm}       
\label{thm61}
The map $\rho\colon\pscript _n\times\pscript _n\to\real$ is a metric on $\pscript _n$.
\end{thm}
\begin{proof}
Clearly, $\rho (x,y)\ge 0$ and $\rho (x,y)=\rho (y,x)$ for all $x,y\in\pscript _n$. Also, $\rho (x,x)=0$. Suppose $x,y\in\pscript _n$ with
$x\ne y$. If $\xhat '=x$, $\yhat '=y$, then $x'\ne y'$ so $\rho (x',y')>0$. It follows that $\rho (x,y)>0$. For $x,y,z\in\pscript _n$, if $x=y$ we have that
\begin{equation*}
\rho (x,y)=0\le\rho (x,z)+\rho (z,y)
\end{equation*}
so suppose that $x\ne y$. If $z=x$, then
\begin{equation*}
\rho (x,z)+\rho (z,y)=\rho (x,y)
\end{equation*}
and a similar result holds if $z=y$. Finally, suppose that $x,y,z$ are all distinct. Now there exists $x',y',z'\in\pscript '_n$ with
$\xhat '=x$, $\yhat '=y$, $\zhat '=z$ and $\rho (x,y)=\rho (x',y')$. Since $\rho$ is a metric on $\pscript '_n$ we have that
\begin{equation*}
\rho (x,y)=\rho (x',y')\le\rho (x',z')+\rho (z',y')\le\rho (x,z)+\rho (z,y)
\end{equation*}
Hence, the triangle equality holds so $\rho$ is a metric on $\pscript _n$.
\end{proof}

We now define metrics on $\Omega '_n$ and $\Omega _n$. For $\omega ,\omega '\in\Omega '_n$ with
$\omega =\omega _2\omega _2\cdots\omega _n$, $\omega '=\omega '_1\omega '_2\cdots\omega '_n$ define
\begin{equation}         
\label{eq61}
\rho (\omega ,\omega ')=\max\brac{\rho (\omega _i,\omega '_i)\colon i=1,2,\ldots ,n}
\end{equation}

\begin{thm}       
\label{thm62}
The map $\rho\colon\Omega '_n\times\Omega ' _n\to\real$ is a metric.
\end{thm}
\begin{proof}
Clearly, $\rho (\omega ,\omega ')\ge 0$, $\rho (\omega ,\omega ')=\rho (\omega ',\omega )$ and $\rho (\omega ,\omega ')$=0 if and only if $\omega =\omega '$. If $\omega ''=\omega ''_i\omega ''_2\cdots\omega ''_n\in\Omega '_n$ we have that $\rho (\omega ,\omega ')=\rho (\omega _j,\omega '_j)$ for some $j\in\brac{1,\ldots ,n}$ and that
\begin{equation*}
\rho (\omega ,\omega ')=\rho (\omega _j,\omega '_j)
  \le\rho (\omega _j,\omega ''_j)+\rho (\omega ''_j,\omega '_j)\le\rho (\omega ,\omega '')+\rho (\omega '',\omega ')\qedhere
\end{equation*}
\end{proof}

We can define a metric on $\Omega _n$ in a similar way using \eqref{eq61}. Other metrics on $\Omega '_n$ and $\Omega _n$ that might be convenient are
\begin{align*}
\rho _1(\omega ,\omega ')&=\sum _{i=1}^n\rho (\omega _i,\omega '_i)\\
\rho _2(\omega ,\omega ')&=\sqbrac{\sum _{i=1}^n\rho (\omega _i,\omega '_i)^2}^{1/2}
\end{align*}


\begin{thebibliography}{99}
\bibitem{djs10}F.~Dowker, S.~Johnston and S.~Surya, On extending the quantum measure, arXiv: gr-qc 1204.4596 (2010).
\bibitem{gud11}S.~Gudder, Discrete quantum gravity, arXiv: gr-qc 1108.2296 (2011).
\bibitem{gud13a}S.~Gudder, A dynamics for discrete quantum gravity, arXiv: gr-qc 1303.0433 (2013).
\bibitem{gud13b}S.~Gudder, An approach to discrete quantum gravity, arXiv: gr-qc 1305.5184 (2013).
\bibitem{hen09}J.~Henson, Quantum histories and quantum gravity, arXiv: gr-qc 0901.4009 (2009).
\bibitem{rs00}D.~Rideout and R.~Sorkin, A classical sequential growth dynamics for causal sets, \textit{Phys. Rev.} \textbf{D61} (2000),
024002.
\bibitem{sor94}R.~Sorkin, Quantum mechanics as quantum measure theory, \textit{Mod. Phys. Letts.} A9 (1994), 3119--3127.
\bibitem{sor03}R.~Sorkin, Causal sets: discrete gravity, arXiv: gr-qc 0309009 (2003).
\bibitem{sur11}S.~Surya, Directions in causal set quantum gravity, arXiv: gr-qc 1103.6272 (2011).

\end{thebibliography}
\end{document}